\documentclass[11pt]{amsart}

\usepackage{amssymb,amsthm,amsmath,amssymb,xcolor}
\usepackage{enumerate}
\usepackage[hidelinks]{hyperref}

\newcommand{\dd}{\mathrm{d}}
\newcommand{\E}{\mathbb{E}}
\newcommand{\1}{\textbf{1}}
\newcommand{\R}{\mathbb{R}}

\newcommand{\Ent}[1]{h\left( #1 \right)}

\DeclareMathOperator{\Var}{Var}

\usepackage{color}
\def\red{}

\usepackage[paper=a4paper, left=1.3in, right=1.3in, top=1.2in, bottom=1.2in]{geometry}
\newtheorem{theorem}{Theorem}
\newtheorem{lemma}[theorem]{Lemma}
\newtheorem{corollary}[theorem]{Corollary}
\newtheorem{proposition}[theorem]{Proposition}

\theoremstyle{remark}
\newtheorem{remark}[theorem]{Remark}

\theoremstyle{definition}
\newtheorem{conjecture}{Conjecture}

\begin{document}

\title{Sharp moment-entropy inequalities and capacity bounds for {\red symmetric} log-concave distributions}
\author{Mokshay Madiman}
\address{University of Delaware}
\email{madiman@udel.edu}

\author{Piotr Nayar}
\address{University of Warsaw}
\email{nayar@mimuw.edu.pl}

\author{Tomasz Tkocz}
\address{Carnegie Mellon University}
\email{ttkocz@math.cmu.edu}

\thanks{M.M. was supported in part by the U.S. National Science Foundation through the grant DMS-1409504. 
P.\!~N. was partially supported by the National Science Centre Poland grants
2015/18/A/ST1/00553 and 2018/31/D/ST1/01355. The research leading to these results is part of a project that has received funding from the 
European Research Council (ERC) under the European Union's Horizon 2020 research and innovation programme (grant agreement No 637851).
T.T. was supported in part by the Simons Foundation Collaboration grant and the U.S. National Science Foundation through the grant DMS-1955175.
This work was also supported by the NSF under Grant No. 1440140, while the authors were in residence at the Mathematical Sciences Research Institute 
in Berkeley, California, for the ``Geometric and Functional Analysis'' program during the fall semester of 2017. A preliminary version \cite{MNT19:isit}
of this paper was presented at ISIT 2019.}

\begin{abstract}
We show that the uniform distribution minimizes entropy among all one-dimensional symmetric log-concave distributions with fixed variance,
as well as various generalizations of this fact to R\'enyi entropies of orders less than 1 and with moment constraints involving $p$-th absolute
moments with $p\leq 2$. 
As consequences, we give new capacity bounds for additive noise channels with symmetric log-concave noises, as well as for timing channels
involving positive signal and noise where the noise has a decreasing log-concave density.
In particular, we show that the capacity of an additive noise channel with  symmetric, log-concave noise
under an average power constraint is at most 0.254 bits per channel use
greater than the capacity of an additive Gaussian noise channel with the same noise power.
Consequences for reverse entropy power inequalities and connections to the slicing problem in convex geometry are also discussed.
\end{abstract}

\maketitle

{\footnotesize

}
\bigskip

\section{Introduction}

It is a classical fact going back to Boltzmann \cite{Bol1896} that when the variance of a real-valued
random variable $X$ is kept fixed, the differential entropy is maximized by taking $X$ to be Gaussian.
As is standard in information theory, we use the definition of Shannon \cite{Sha48}: the
differential entropy (or simply {\it entropy}, henceforth, since we have no need to deal with
discrete entropy in this note) of a random vector $X$ with density $f$ is defined as
\[
\Ent{X} = {\red \Ent{f} =} -\int_{\R^n} f \log f,
\]
provided that this integral exists, this definition having a minus
sign relative to Boltzmann's $H$-functional.
It is easy to see that if one tried to {\it minimize} the entropy instead of maximizing it, there is no
minimum among random variables with densities-- indeed, {\red appropriately approximating a discrete random variable with variance $1$ with continuous random  variables yields a sequence of probability densities with bounded variance which would have differential entropies converging to $-\infty$ (say, $f_\epsilon(x) = \frac{1}{4\epsilon}\1_{(-1-\epsilon,-1+\epsilon)}(x) + \frac{1}{4\epsilon}\1_{(1-\epsilon,1+\epsilon)}(x)$, $x \in \R$, $\epsilon\to 0$)}. Nonetheless,
it is of significant interest to identify minimizers of entropy within structured subclasses
of probability measures. 
For instance, it was observed independently by Keith Ball (unpublished)
and in \cite{BM11:it} that the question of minimizing entropy under a covariance matrix constraint within
the class of log-concave measures on $\R^n$ is intimately tied to the well known hyperplane
or slicing conjecture in convex geometry.

More generally, log-concave distributions emerge naturally from the interplay between 
information theory and convex geometry, and have recently
been a very fruitful and active topic of research (see the recent survey \cite{MMX17:0}). 
A probability density $f$ on $\R$ is said to be {\it log-concave}
if it is of the form $f = e^{-V}$ for a convex function $V: \R \to \R\cup\{\infty\}$. 
It is said to be {\it symmetric} if $f(-x)=f(x)$ for each $x\in\R$.
We note that these are very natural assumptions to place on the noise distribution
for an additive noise channel; the ubiquitous centered Gaussian noise is clearly symmetric and log-concave,
and expanding our focus to the infinite-dimensional class of symmetric, log-concave distributions
allows for much more flexibility in modeling the noise while preserving the broad
qualitative features (such as unimodality and symmetry) of the standard Gaussian noise.

Our primary goal in this note is to establish some sharp inequalities relating the entropy 
(and in fact, a more general class of R\'enyi entropies)
to moments for symmetric, log-concave distributions. 
For the sake of simplicity, we present in this introduction only the result for Shannon differential entropy.
Our main result shows that among all symmetric log-concave probability distributions on $\R$ with fixed variance, 
the uniform distribution has minimal entropy. In fact, we obtain a slightly more general result involving the $p$-th moments
for $p\leq 2$.  Let us use $\sigma_p(X)$ to denote $(\E \big[ |X|^p \big])^{1/p}$.

\begin{theorem}\label{thm:ent-pmom}
Let $X$ be a symmetric log-concave random variable and $p \in (0,2]$. Then,
\[
\Ent{X} \geq \log \sigma_p(X)+ \log\left[2(p+1)^{1/p}\right],
\]
with equality if and only if $X$ is a uniform random variable.
\end{theorem}

It is instructive to write this inequality using the {\it entropy power} of $X$, defined by
$$
\mathcal{N}(X)=\frac{1}{2\pi e}e^{2\Ent{X}},
$$ 
in which case it becomes
$$
\mathcal{N}(X) \geq \frac{2}{\pi e}(p+1)^{2/p} \sigma_p(X)^{2} .
$$
In the special case $p=2$ corresponding to the variance, we have the sandwich inequality
$$
\frac{6}{\pi e} \Var(X)\leq \mathcal{N}(X)\leq \Var(X) ,
$$
with both inequalities being sharp in the class of symmetric log-concave random variables (the one on the left, 
coming from Theorem~\ref{thm:ent-pmom}, giving equality uniquely for the uniform distribution, while the one on
the right, coming from the maximum entropy property of the Gaussian, giving equality uniquely for the Gaussian distribution.)
Note that $6(\pi e)^{-1}\approx 0.7026$, so the range of entropy power given variance is quite constrained for symmetric log-concave random variables. 

Theorem~\ref{thm:ent-pmom} can be viewed as a sharp version in the symmetric case of some of the 
estimates from \cite{BM11:it, MK18b}  (see also \cite{CSH17} for upper bounds on the variance 
in terms of the entropy for mixtures of densities of the form $e^{-|t|^\alpha}$). However, finding the sharp version is quite delicate and one needs significantly 
more sophisticated methods. Our argument comprises two main steps: first, we reduce the problem 
to \emph{simple} random variables (compactly supported, piecewise exponential density), 
using ideas and techniques developed by Fradelizi and Guedon \cite{FG06} in order to elucidate the 
sophisticated localization technique of Lov\'asz and Simonovits \cite{LS93},
and second, we prove a nontrivial two-point inequality in order to verify the inequality for such random variables.

This note is organized as follows. In Section~\ref{sec:cap}, we discuss the implications of Theorem~\ref{thm:ent-pmom}
for bounds on the capacity of two classes of additive noise channels. First, 
for channels with average-power-constrained, real-valued signals and symmetric, log-concave noise,
we show that the capacity is at most $\frac{1}{2}\log_2 \big(\frac{\pi e}{6}\big)\approx 0.254$ bits per channel use
greater than the capacity of an additive Gaussian noise channel with the same noise power.
Second, for channels with positive, average-amplitude-constrained signals and decreasing log-concave (positive) noise,
we show that the capacity is at most $\log_2 (e/2)\approx 0.443$ bits per channel use
greater than the capacity of an additive exponential noise channel with the same mean.
In Section~\ref{sec:slicing}, we explain the connection of  Theorem~\ref{thm:ent-pmom}
to the famous ``slicing problem'' in asymptotic convex geometry, and also remark on an 
entropy interpretation of a classical lemma (which also
plays a role in our proof of Theorem~\ref{thm:ent-pmom}). 
Section~\ref{sec:proof-ent-vs-var} contains a full proof of Theorem~\ref{thm:ent-pmom}, while Section~\ref{sec:renyi}
develops  an extension of Theorem~\ref{thm:ent-pmom} to a class of R\'enyi entropies as well as various consequences of this extension,
including to reverse entropy power inequalities.
In Section~\ref{sec:cvx-cap}, we make a very general observation about the capacity of additive noise
channels (not just for real signals and noise, but applying to a very general class of groups)--
namely, that the capacity is convex as a function of the noise distribution,
and discuss its implications.
We conclude in Section~\ref{sec:disc} with various remarks that yield some additional insight
into our results, as well as why the question of finding the best additive noise channel among
symmetric log-concave noises is non-trivial.

\section{Channel capacity bounds}
\label{sec:cap}

\subsection{Additive noise channels with real signals}
\label{ss:anc-real}

Let $X, Y$ be random vectors taking values in $\R^n$, 
with probability density functions $f$ and $u$ respectively. 
The {\it relative entropy} between $f$ and $u$ is, as usual, defined by
$$
D(X\|Y)=D(f\|u):=\int_{\R^n} f(x) \log\frac{f(x)}{u(x)} dx ,
$$
and is always nonnegative (though possibly $+\infty$).
For a random vector $X$ with density $f$ on $\R^n$ that has finite second moment (or covariance matrix), 
the {\it relative entropy from Gaussianity} is defined by
$$
D(X)=D(f):= \inf_{u} D(f\|u) = D(f\|g) ,
$$
where the infimum is taken over all Gaussian densities $u$ on $\R^n$,
and is achieved by the Gaussian density $g$ with the same mean
and covariance matrix as $X$. If $Z$ has density $g$,
then it is a classical and easy observation (see, e.g., \cite{CT06:book}) that 
$$
D(X)=h(Z)- h(X),
$$
or equivalently, $D(f)=D(f\|g)=h(g)-h(f)$. In particular, this implies that the Gaussian
is the unique maximizer of entropy when the mean and covariance matrix are fixed.

Consider the memoryless channel additive  noise $N$ 
that takes in real signals and has a power budget $P$, i.e.,
the channel can transmit in blocks any codeword $(x_1, \ldots, x_n)\in\R^n$ that
satisfies the average power constraint
$$
\frac{1}{n}\sum_{i=1}^n x_i^2 \leq P.
$$
The output produced by the channel at the receiver when $X$ is the input is
$Y=X+N$, where the noise $N$ is independent of $X$. 
Let $C_P(N)$ be the capacity of this channel, i.e., the supremum of achievable
rates (measured in bits per channel use) that can be transmitted across the channel with the receiver
being able to decode the transmitted message with vanishing error probability as block length grows.
From the classical channel coding theorem of Shannon \cite{Sha48}, we know that
$$
C_P(N)=\sup_{X: E|X|^2=P} I(X;Y) = \sup_{X: E|X|^2=P} h(X+N)-h(N) ,
$$
where $I(X;Y)$ as usual denotes the mutual information between $X$ and $Y$.
In fact, in his original paper, Shannon \cite{Sha48} not only determined the capacity of the
AWGN (additive white Gaussian noise) channel, but also formulated bounds on the capacity
when the additive noise is not Gaussian. Specifically,  \cite[Theorem 18]{Sha48} asserts that
\begin{equation}\label{eq:shannon-cap-bd}
\frac{1}{2}\log \bigg(1+\frac{P}{\mathcal{N}(N)}\bigg) \leq {\red C_P(N)} \leq \frac{1}{2}\log \bigg(\frac{P+{\red \Var(N)}}{\mathcal{N}(N)}\bigg) ,
\end{equation}
with $\mathcal{N}(N)$ being the entropy power of the noise.
The upper bound just uses the fact that the Gaussian maximizes entropy under a second moment 
constraint, while the lower bound is a simple application of {\red the Shannon-Stam entropy power inequality \cite{Sha48, Sta59}, which asserts that
\begin{equation}\label{eq:EPI}
\mathcal{N} (X+Y) \geq  \mathcal{N} (X) +  \mathcal{N} (Y),
\end{equation}
for any two independent random vectors $X$ and $Y$ in $\R^n$ for which
the three entropies in the inequality are defined (see \cite{BC15:1} for a discussion of why
just existence of $N(X)$ and $N(Y)$ is insufficient). }

A consequence of the lower bound 
in \eqref{eq:shannon-cap-bd} is that the ``worst'' additive noise is Gaussian,
in the sense that for fixed noise power, Gaussian noise minimizes capacity.
Indeed, if $Z$ is Gaussian noise with
$\Var(Z)=\Var(N)=P_N$, then 
$$
C_P(Z)=\frac{1}{2}\log \bigg(1+\frac{P}{P_N}\bigg)
=\frac{1}{2}\log \bigg(1+\frac{P}{\mathcal{N}(Z)}\bigg)
\leq \frac{1}{2}\log \bigg(1+\frac{P}{\mathcal{N}(N)}\bigg) 
\leq C_P(N).
$$
On the other hand, a consequence of the upper bound in \eqref{eq:shannon-cap-bd} is that
\begin{equation*}\begin{split}
C_P(N)&\leq C_P(Z)+\frac{1}{2}\log \bigg(\frac{P+P_N}{\mathcal{N}(N)}\bigg)- \frac{1}{2}\log \bigg(1+\frac{P}{P_N}\bigg)\\
&= C_P(Z)+\frac{1}{2}\log\frac{P_N}{\mathcal{N}(N)}\\
&= C_P(Z)+ h(Z)-h(N)\\
&=C_P(Z)+D(N) ,
\end{split}\end{equation*}
where $D(N)$ is the relative entropy of {\red $N$} from Gaussianity.
We summarize these observations, of which Ihara \cite{Iha78} developed multidimensional  and continuous-time extensions, 
in the proposition below. 

\begin{proposition}\label{prop:ihara}\cite{Sha48, Iha78}
Let $C_P(N)$ be the capacity of the additive noise channel with a noise $N$ of finite variance and
input signal power budget of $P$.
If $Z$ is a Gaussian random variable with mean 0 and variance equal to that of $N$, then
$$
C_P(Z)\leq C_P(N)\leq C_P(Z)+D(N) .
$$
\end{proposition}

Let us note that we may interpret Theorem~\ref{thm:ent-pmom} as a bound on relative entropy.
Specifically, we can rewrite Theorem~\ref{thm:ent-pmom} (for $p=2$)  as follows.

\begin{corollary}\label{cor:iso}
If the random variable $N$ has a symmetric, log-concave distribution, then
$$
D(N)\leq \frac{1}{2}\log \bigg(\frac{\pi e}{6}\bigg) ,
$$
with equality if and only if $N$ is uniformly distributed on an interval.
\end{corollary}

Combining Corollary~\ref{cor:iso} and Proposition~\ref{prop:ihara}, we obtain the following corollary.

\begin{corollary}\label{cor:cap}
If the random variable $N$ has a symmetric, log-concave distribution, then
$$
C_P(N)\leq C_P(Z)+ \frac{1}{2}\log \bigg(\frac{\pi e}{6}\bigg).
$$
\end{corollary}

Corollary~\ref{cor:cap} implies that an additive noise channel with symmetric, log-concave noise
has capacity that is at most $\frac{1}{2}\log_2 \big(\frac{\pi e}{6}\big)\approx 0.254$ bits per channel use 
greater than the capacity of an AWGN channel with the same noise power. 
(We remark that for most inequalities in this
paper, the logarithms may be taken to an arbitrary base, as long as entropies, capacities, and
other information functionals also use the same base in their definitions. Therefore
to get numbers in units of bits rather than nats, all we need to do is to take the logarithm to base 2.)

We can, in fact, say more. Let us define the {\it restricted capacity}
of an additive noise channel with noise $N$ and input power constraint $P$ by
$$
C_P^{LC}(N):= \sup_{X} I(X;X+N) ,
$$
where the supremum is taken over all  symmetric and log-concave distributions for $X$ such that $\E X^2 \leq P$.

If $N$ is a symmetric and log-concave noise, using the fact that 
$I(X;X+N)=h(X+N)-h(N)$ and that $X+N$ is symmetric and log-concave when both $X$ and $N$ are,
and applying Theorem~\ref{thm:ent-pmom}, we have
\begin{equation*}\begin{split}
C_P^{LC}(N)&\,\geq \, \frac{1}{2} \log[ 12(P+P_N)]  -h(N)
\, =\, \frac{1}{2} \log\bigg[ \frac{12(P+P_N)}{2\pi e P_N}\bigg]  +D(N)\\
&\,=\frac{1}{2} \log\bigg( \frac{6}{\pi e}\bigg) + C_P(Z) +D(N)
\, \geq\, C_P(N) -\frac{1}{2} \log\bigg( \frac{\pi e}{6}\bigg) ,
\end{split}\end{equation*}
where we used Proposition~\ref{prop:ihara} for the last inequality.
Thus we obtain:

\begin{corollary}\label{cor:cap-loss}
If $N$ is a symmetric and log-concave noise,
$$
0\leq C_P(N)-C_P^{LC}(N) \leq \frac{1}{2} \log\bigg( \frac{\pi e}{6}\bigg).
$$
\end{corollary}

Corollary~\ref{cor:cap-loss} says that the loss in capacity from ``restricting'' the input distribution 
to symmetric log-concave distributions is at most $\frac{1}{2}\log_2 \big(\frac{\pi e}{6}\big)\approx 0.254$ 
bits per channel use. We note the curious appearance of the same constant as in Corollary~\ref{cor:cap},
which may suggest that an inequality holds between $C_P^{LC}(N)$ and $C_P(Z)$, but this remains
unknown.

\subsection{Additive noise channels with nonnegative signal and noise}
\label{ss:anc-pos}

Motivated by applications involving timing, such as telephone signalling or trying to send bits through queues \cite{AV96},
it is of interest to consider additive noise channels where both signal and noise are nonnegative.
Other aspects of such channels have been considered in \cite{WA05, OB06, BGN17}.

The typical setup is as follows. Messages can be encoded in blocks 
using any codeword $(x_1, \ldots, x_n)\in\R_+^n$ that satisfies the constraint
$$
\frac{1}{n}\sum_{i=1}^n x_i \leq P.
$$
The output produced by the channel at the receiver when $X$ is the input is
$Y=X+N$, where the nonnegative noise $N$ is independent of $X$. 
We call this channel the ``positive additive noise channel'' (PANC)
with noise $N$ and budget $P$, and denote its capacity by  $C_P^{+}(N)$.

Verd\'u \cite{Ver96} considers such channels  and proves the following pleasing result reminiscent 
of Proposition~\ref{prop:ihara}.

\begin{theorem} \cite{Ver96}\label{thm:verdu}
Let  $N$ be a positive random variable with mean $a$.
If $E$ denotes an exponential random variable with mean $a$, then
$$
C_P^{+}(E) \leq C_P^{+}(N) \leq C_P^{+}(E)+D(N\|E) .
$$
\end{theorem}

The first inequality says that exponential noise is the worst noise for fixed budget $P$. 
Note that $C_P^{+}(E)=\log (1+\frac{P}{a})$ has an explicit formula. 

Our main theorem implies the following for densities on the positive real line.

\begin{theorem}\label{thm:pos-main}
If $f$ is a non-increasing, log-concave density on $(0,\infty)$, and $p\in (0,2]$, then
$$
h(f) \geq  \frac{1}{p} [ \log \E(X^p) + \log (p+1)].
$$
\end{theorem}

\begin{proof}
Suppose $f$ is a non-increasing, log-concave density on $(0,\infty)$. Then 
$$
f_{sym}(x)=\frac{f(x)1_{x>0}+f(-x) 1_{x<0}}{2}
$$
is a symmetric, log-concave density on $\mathbb{R}$. We have
\begin{eqnarray*}\begin{split}
h(f_{sym})&=-\int_{\mathbb{R}} f_{sym}(x) \log f_{sym}(x)  \dd x \\
&= -2\int_0^\infty \frac{f(x)}{2} \log \bigg(\frac{f(x)}{2} \bigg) \dd x \\
&= h(f) + \log 2
\end{split}\end{eqnarray*}
and
\begin{eqnarray*}\begin{split}
\sigma_p(f_{sym})^p &= \int_{\mathbb{R}} |x|^p f_{sym}(x) \dd x\\
&= \frac{1}{2} \int_{-\infty}^0 (-x)^p f(-x) \dd x + \frac{1}{2} \int_0^{\infty} x^p f(x) \dd x  \\
&= \int_0^{\infty} x^p f(x) \dd x = {\red \E(X^p)}.
\end{split}\end{eqnarray*}
Since from Theorem~\ref{thm:ent-pmom},
$\Ent{f_{sym}} \geq \log \sigma_p(f_{sym})+ \log\left[2(p+1)^{1/p}\right]$,
we conclude that
$$
h(f) \geq \frac{1}{p}\log \E(X^p)+ \log\left[2(p+1)^{1/p}\right] -\log 2
= \frac{1}{p} [ \log \E(X^p) + \log (p+1)]
$$
for $p\in (0,2]$, yielding the desired inequality.
\end{proof}

Since entropy is maximized at the exponential distribution under a mean
constraint, we may rewrite the conclusion of Theorem~\ref{thm:pos-main} as
$$
D(X\|E)=h(E)-h(X) \leq \log (ea) - \frac{1}{p}[ \log \E(X^p) + \log (p+1)],
$$
when $E$ is an exponential random variable with the same mean $a$ as $X$.
When specialized to $p=1$, this reads
$$
D(X\|E)\leq \log (ea)-\log a - \log 2
= \log \big(\frac{e}{2}\big).
$$

Combining the preceding inequality with Theorem~\ref{thm:verdu} gives the following.

\begin{corollary}\label{thm:exp-gap}
If $N$ is a distribution on the positive real line with non-increasing, log-concave density, then
$C_P^{+}(N)\leq C_P^{+}(E)  + \log \big(\frac{e}{2}\big)$.
\end{corollary}

Corollary~\ref{thm:exp-gap} tells us that a PANC with any decreasing, log-concave noise 
has a capacity that is at most $\log_2 (e/2)\approx 0.443$ bits per channel use 
more than the PANC with exponential noise of the same mean.

\section{Relation to the slicing problem in convex geometry}
\label{sec:slicing}

For any probability density function $f$ on $\R^n$ with covariance matrix $R$,
define its {\it isotropic constant} $L_f$ by
$$
L_f^2=\|f\|_{\infty}^{2/n} (\det(R))^{\frac{1}{n}} .
$$
The isotropic constant has a nice interpretation for uniform distributions
on convex sets $K$. If one rescales $K$ (by a linear
transformation) so that the volume of the convex set is 1 and the covariance matrix is
a multiple of the identity, then $L_K^2:=L_f^2$ is the value of the
multiple.

Observe that both $D(f)$ and $L_f$ are affine invariants. Their relationship was made explicit in \cite[Theorem V.1]{BM11:it}.

\newcommand{\vol}{\text{Vol}}
\begin{theorem}\cite{BM11:it}\label{thm:D-iso}
For any density $f$ on $\R^n$,
$$
\frac{1}{n} D(f) \leq \log [\sqrt{2\pi e} L_f ] ,
$$
with equality if and only if $f$ is the uniform density on some
set of positive, finite Lebesgue measure.
If $f$ is a log-concave density on $\R^n$, then
$$
\log \bigg[\sqrt{\frac{2\pi}{e}} L_f \bigg] \leq \frac{1}{n} D(f) ,
$$
with equality if $f$ is a product of one-dimensional 
exponential densities.
\end{theorem}

Since $D(f)\geq 0$, Theorem~\ref{thm:D-iso} immediately yields $\sqrt{2\pi e} L_f \geq 1$,
which is the optimal dimension-free lower bound on isotropic constants.
On the other hand, the problem of whether the isotropic constant
is bounded from above by a universal constant for the class of uniform distributions on symmetric convex bodies, 
which was first raised by Bourgain \cite{Bou86} in 1986 (see also \cite{Bal88, MP89}),
remains open. 

\begin{conjecture}\cite{Bou86}\label{conj:hyp1}[{\sc Slicing Problem or Hyperplane Conjecture}] 
 There exists a universal, positive constant $c$ (not depending on $n$) 
such that for any symmetric convex set $K$ of unit volume in $\mathbb{R}^n$,  there exists a hyperplane $H$
such that the $(n-1)$-dimensional volume of the section
$K\cap H$ is bounded below by $c$. 
\end{conjecture}

 The slicing problem has spurred a large literature, a synthesis of which may be found in the book \cite{BGVV14:book}.
 For our purposes, we note that there are several equivalent formulations of the conjecture, all of a geometric or functional
analytic flavor. Motivated by a seminal result of Hensley \cite{Hen80} (cf. \cite{MP89})
that 
$c_1 \leq L_K \vol_{n-1}(K\cap H) \leq c_2$,
for any isotropic convex body $K$ in $\R^n$ and any hyperplane $H$ passing through
its barycenter (with $c_2>c_1>0$ being universal constants), it can be shown that
the hyperplane conjecture is equivalent to the statement that
the isotropic constant of a symmetric convex body in $\R^n$ is bounded from above by a universal constant (independent of $n$).
Furthermore, it turns out that the conjecture is also equivalent to the statement that
the isotropic constant of a symmetric log-concave density in $\R^n$ is bounded 
from above by a universal constant independent of dimension.
Moreover, the assumption of central symmetry may be removed from
the conjecture if it is true \cite{MP99}, but we focus on symmetric bodies and densities in this note.

Using this formulation in terms of isotropic constants and Theorem~\ref{thm:D-iso}, \cite{BM11:it} proposed
the following ``entropic form of the hyperplane conjecture'': 
For any symmetric log-concave density $f$ on $\R^n$ and some universal
constant $c$,
$\frac{D(f)}{n}\leq c$.
Thus the conjecture is a statement about the (dimension-free) closeness of an 
arbitrary symmetric log-concave measure to a Gaussian measure.

Existing partial results on the slicing problem already give insight
into the closeness of log-concave measures to Gaussian measures. 
While there are a string of earlier results (see, e.g., \cite{Bou91, Dar95, Pao00}),
the current best bound, obtained by Klartag \cite{Kla06} (cf., \cite{LV17:focs}),
asserts that $L_K \leq cn^{1/4}$.  Using a transference result of Ball \cite{Bal88} from
convex bodies to log-concave functions, the same bound is seen to also apply 
to $L_f$, for a general log-concave density $f$.
Combining this with Theorem~\ref{thm:D-iso}  leads immediately to the conclusion
that for  any log-concave density $f$ on $\R^n$,
$D(f)\leq \frac{1}{4} n\log n + cn$,
for some universal constant $c>0$. 

The original motivation for our exploration of Theorem~\ref{thm:ent-pmom} actually arose
from the hyperplane conjecture: our hope was to understand the extremizers 
(for the formulations in terms of relative entropy and the isotropic constant) 
in low dimensions as a source of intuition. Corollary~\ref{cor:iso} speaks to this
question in dimension 1 for the class of symmetric log-concave densities (of course,
in dimension 1, the geometric question for convex sets is trivial since there is only one
convex set up to scaling in $\R$). 
Indeed, Corollary~\ref{cor:iso} implies for any  symmetric, log-concave density $f$ on $\R$, 
$L_f \leq  \frac{e}{\sqrt{12}}$.
Since the uniform is not an extremizer for the 
upper bound on $L_f$  in terms of $D(f)$ (though the symmetrized exponential is),
this bound is not sharp. Nonetheless, let us observe that a sharp bound on the isotropic constant
in dimension 1 is actually implied by Lemma~\ref{lm:rev-Hensley} below.
Indeed, Lemma~\ref{lm:rev-Hensley} (or
the equivalent Proposition~\ref{prop:rev-Hensley}) implies that
in the class of symmetric, log-concave densities on $\R$, $L_f\leq \frac{\Gamma(3)}{4}=\frac{1}{2}$,
with equality if and only if $f$ is a symmetrized exponential density.
It is interesting to note that, already in dimension 1, the extremizers for the isotropic constant formulation of the slicing problem
are different from those for the relative entropy formulation of it.

As briefly mentioned earlier, the questions discussed in this section are of interest both with and without the central symmetry assumption.
Our main result does, in fact, provide a bound even in the non-symmetric case,
thanks to the observation of \cite{BM13:goetze} that $\mathcal{N}(X-Y)\leq e^2 \mathcal{N}(X)$ 
if $X, Y$ are i.i.d. with a log-concave distribution on $\R^n$.
(The constant, which is not sharp, is conjectured in \cite{MK18} to be 4 and to be achieved by the product distribution 
whose 1-dimensional marginals are the exponential distribution.)
This immediately implies, from the fact that $X-Y$ has a symmetric, log-concave distribution, that
$$
\mathcal{N}(X)\geq \frac{\mathcal{N}(X-Y)}{e^2} \geq \frac{6}{\pi e^3} \Var(X-Y) = \frac{12}{\pi e^3} \Var(X) \approx 0.19 \Var(X) .
$$
However, this bound is significantly inferior to \cite[Theorem 3]{MK18b}, which shows that $N(X)\geq 4\Var(X)$.
When translated to bounds on $D(Y)$, this bound of Marsiglietti and Kostina \cite{MK18b} reads 
as $D(Y)\leq \frac{1}{2}\log \big(\frac{\pi e}{2}\big)$ for any log-concave density on $\R$ (not necessarily symmetric). 
While this bound improves on an earlier bound of $\frac{1}{2}\log(\pi e)$ obtained by \cite{BM11:it}, it remains
suboptimal for the class of log-concave distributions. We believe that the optimal bound
on $D(Y)$ for log-concave random variables $Y$ that are not necessarily symmetric should be
$\frac{1}{2}\log \big(\frac{2\pi}{e}\big)$, which is achieved for the exponential distribution with density $e^{-x}$ supported on the positive real line,
but we have been unable to prove this so far.



We have the following sharp relation between moments and the maximum value of a symmetric, log-concave function on the real line. 

\begin{lemma}\label{lm:rev-Hensley}
For every even log-concave function $f:\R\to [0,+\infty)$, we have
\[
f(0)^p\int |x|^pf(x) \dd x \leq 2^{-p}\Gamma(p+1)\left(\int f(x) \dd x\right)^{p+1}.
\]
Equality holds if and only if $f(x)=ce^{-C|x|}$ for some positive constants $c, C$.
\end{lemma}

\begin{proof}
By homogeneity we can assume that $f(0) = 1$. Consider $g(x) = e^{-a|x|}$ such that $\int g = \int f$. By log-concavity, 
there is exactly one sign change point $x_0$ for $f - g$. We have
\[
\int |x|^p[f(x) - g(x)] = \int [|x|^p-|x_0|^p][f(x) - g(x)] \leq 0 ,
\]
since the integrand is nonpositive. It remains to verify the lemma for $g$, which holds with equality.
\end{proof}

The inequality in the lemma is not new; indeed, it follows from classical and more general reverse H\"older inequalities
independently discovered by Ball \cite[Lemma 4]{Bal88} and Milman-Pajor \cite[Lemma 2.6]{MP89} (see also \cite{BFLM17}).
Moreover the idea of the proof involving sign changes has also found use in recent investigation of moment sequences
of symmetric, log-concave densities \cite{ENT18:2}.  

Observe that since $f(0)=\max_x f(x)=\|f\|_\infty$ for a symmetric, log-concave density $f$, Lemma~\ref{lm:rev-Hensley} 
may be rewritten using the language of R\'enyi entropy.

\begin{proposition}\label{prop:rev-Hensley}
If $X$ has a symmetric, log-concave density $f$ on $\R$, we have
\begin{equation}\label{eq:rev-h-ent}
h_\infty(X)\geq \log \sigma_p(X)+ \frac{1}{p} \log \bigg[ \frac{2^p}{\Gamma(p+1)} \bigg] ,
\end{equation}
with equality if and only if $X$ has a symmetrized exponential distribution, i.e., $f(x)=\frac{c}{2} e^{-c|x|}$ for some $c>0$. 
\end{proposition}

If we tried to use Proposition~\ref{prop:rev-Hensley} to get a bound on entropy using the fact that $h(X)\geq h_\infty(X)$,
it would not be sharp since the former inequality is sharp only for  symmetrized exponentials,
and the latter is sharp only for uniforms. Consequently we need a different technique to prove
Theorem~\ref{thm:ent-pmom}. The approach we use in the next section utilizes the concavity property of the Shannon entropy $h$, 
which does not hold for $h_\infty$.


\section{Proof of Theorem \ref{thm:ent-pmom}}
\label{sec:proof-ent-vs-var}

Let $\mathcal{F}$ be the set of all even log-concave probability density functions on $\R$. Define for $f \in \mathcal{F}$ the following functionals: entropy,
\[
\Ent{f} = -\int f \log f,
\]
and $p$-th moment,
\[
\sigma_p(f) = \left(\int |x|^pf(x)\dd x\right)^{1/p}.
\]
Our goal is to show that
\[
\inf_\mathcal{F} \Big\{\Ent{f} - \log\left[\sigma_p(f)\right]\Big\} = \log\left[2(p+1)^{1/p}\right].
\]

\subsection*{Reduction}

\subsubsection*{Bounded support}

First we argue that it only suffices to consider compactly supported densities. Let $\mathcal{F}_L$ be the set of all densities from $\mathcal{F}$ which are supported in the interval $[-L,L]$. Given $f \in \mathcal{F}$, by considering $f_L = \frac{f\1_{[-L,L]}}{\int_{-L}^L f}$, which is in $\mathcal{F}_L$, and checking that $\Ent{f_L}$ and $\sigma_p(f_L)$ tend to $\Ent{f}$ and $\sigma_p(f)$, we get
\[
\inf_\mathcal{F} \Big\{\Ent{f} - \log\left[\sigma_p(f)\right]\Big\} = 
\inf_{L > 0} \inf_{\mathcal{F}_L} \Big\{\Ent{f} - \log\left[\sigma_p(f)\right]\Big\}.
\]
This last infimum can be further rewritten as
\[
\inf_{\alpha, L > 0} \Big(\inf \left\{\Ent{f}, \ f \in \mathcal{F}_L, \sigma_p(f) = \alpha\right\}  - \log \alpha\Big).
\]
Consequently, to prove Theorem \ref{thm:ent-pmom}, it suffices to show that for every $\alpha, L > 0$, we have
\[
\inf \left\{\Ent{f}, \ f \in \mathcal{F}_L, \sigma_p(f) = \alpha\right\} \geq \log\alpha + \log\left[2(p+1)^{1/p}\right].
\]

\subsubsection*{Degrees of freedom}

We shall argue that the last infimum is attained at densities $f$ which on $[0, \infty)$ are first constant and then decrease exponentially. Fix positive numbers $\alpha$ and $L$ and consider the set of densities $A = \{f \in \mathcal{F}_L, \sigma_p(f) = \alpha\}$. 

\bigskip
\noindent\underline{Step I.} We show that $M=\sup_{f \in A} -\Ent{f}$ is finite and attained at a point from the set $A$. To see the finiteness we observe that by Lemma \ref{lm:rev-Hensley} for every $f \in A$ we get
\[
	-h(f) = \int f \log f \leq \log(\| f\|_\infty) = {\red \log f(0)} \leq \frac1p \log\left( \frac{2^{-p} \Gamma(p+1)}{\alpha^p} \right).
\]
In order to show that the supremum is attained on $A$, we need the following lemma.

\begin{lemma}\label{lm:subseq}
Let $(f_n)_{n \geq 1}$ be a sequence of functions in $A$. Then there exists a subsequence $(f_{n_k})_{k \geq 1}$ converging pointwise to a function $f$ in $A$.
\end{lemma}

\begin{proof}
As noted above, the functions from $A$ are uniformly bounded (by Lemma \ref{lm:rev-Hensley}) and thus, using a standard diagonal argument, by passing to a subsequence, we can assume that $f_n(q)$ converges for every rational $q$ (in $[-L,L]$), say to $f(q)$. Notice that $f$ is log-concave on the rationals, that is $f(\lambda q_1 + (1-\lambda) q_2) \geq f(q_1)^{\lambda} f(q_2)^{1-\lambda}$, for all rationals $q_1, q_2$ and $\lambda \in [0,1]$ such that $\lambda q_1 + (1-\lambda)q_2$ is also a rational. Moreover, $f$ is even and nonincreasing on $[0,L]$. Let $L_0 = \inf\{q > 0, q \text{ is rational}, f(q) = 0\}$. If $x > L_0$, then pick any rational $L_0< q < x$ and observe that $f_n(x) \leq f_n(q) \to f(q) = 0$, so $f_n(x) \to 0$. The function $f$ is continuous on $[0,L_0)$. If $0 < x < L_0$, consider rationals $q_1, q_2, r$ such that $q_1 < x < q_2 < r < L_0$. Then, by monotonicity and log-concavity,
\[
1 \leq \frac{f(q_1)}{f(q_2)} \leq \left[\frac{f(q_1)}{f(r)}\right]^{\frac{q_2-q_1}{r-q_1}} \leq \left[\frac{f(0)}{f(r)}\right]^{\frac{q_2-q_1}{r-q_1}} \leq \left[\frac{f(0)}{f(r)}\right]^{\frac{q_2-q_1}{r-x}},
\]
thus $\lim_{q_1\to x-} f(q_1) = \lim_{q_2\to x+} f(q_2)$ (these limits exist by the monotonicity of $f$). Now for any $\varepsilon > 0$, take rationals $q_1$ and $q_2$ such that $q_1 < x < q_2$ and $f(q_1)-f(q_2) < \varepsilon$. Since, $f_n(q_2) \leq f_n(x) \leq f_n(q_1)$, we get 
\[
f(q_2) \leq \liminf f_n(x) \leq \limsup f_n(x) \leq f(q_1).
\]
therefore $\limsup f_n(x) - \liminf f_n(x) \leq \varepsilon$. Thus, $f_n(x)$ is convergent, to say $f(x)$. We also set, say $f(L_0) = 0$. Then $f_n$ converges to $f$ at all but two points $\pm L_0$, the function $f$ is even and log-concave. By Lebesgue's dominated convergence theorem, $f \in A$.
\end{proof}

Suppose that $(f_{n})_{n \geq 1}$ is a sequence of elements of $A$  such that $-h(f_n) \to M$. By the lemma, $f_{n_k} \to f$ for some subsequence $(n_k)$ and $f \in A$. By the Lebesgue dominated convergence theorem, $\int f_{n_k}\log f_{n_k} \to \int f\log f$, so $-h(f) = M$. 

\bigskip
\noindent\underline{Step II.} We shall show that $M$ is attained at some extremal point of $A$. Recall that $f \in A$ is called extremal if it is not possible to write $f$ as a combination $f=\lambda f_1 + (1-\lambda)f_2$, where $\lambda \in (0,1)$ and {\red distinct} $f_1, f_2 \in A$. 

Indeed, suppose $f$ is not an extremal point of $A$. Then there exist $\lambda \in (0,1)$ and $f_1, f_2 \in A$ with {\red $f_1 \neq f_2$} such that $f=\lambda f_1 + (1-\lambda)f_2$. Since the entropy functional $-h(f)=\int f\log f$ is strictly convex, we get
\[
	-h(f) = -h(\lambda f_1 + (1-\lambda)f_2) < \lambda (-h(f_1)) + (1-\lambda)(-h(f_2)) \leq M.
\] 
Thus, $-h(f) \ne M$. Since $M$ is attained on $A$, it has to be attained at some extremal point of $A$. 

\bigskip
\noindent\underline{Step III.} Every extremal point in $A$ has at most $2$ degrees of freedom.

Recall the notion of degrees of freedom of log-concave functions introduced in \cite{FG06}, {\red adapted here to even functions}. The degree of freedom of a log-concave {\red even} function $g: \R \to [0,\infty)$ is the largest integer $k$ such that there exist $\delta > 0$ and linearly independent continuous {\red even} functions $h_1,\ldots,h_k$ defined on $\{x \in \R, \ g(x) > 0\}$ such that for every $(\varepsilon_1,\ldots,\varepsilon_k) \in [-\delta,\delta]^k$, the function $g + \sum_{i=1}^k \varepsilon_ih_i$ is log-concave. {\red Let us also notice that an even function is log-concave if and only if its restriction to $[0,\infty)$ is log-concave and non-increasing. Thus in the above definition we could alternatively demand that $g + \sum_{i=1}^k \varepsilon_i h_i$ is log-concave and non-increasing on $[0,\infty)$. Therefore in Step IV below we only consider the restrictions of our functions $g$ to $[0,\infty)$ and consider $h_i$ defined only on this set. Then $h_i$ can be defined on $(-\infty,0)$ via $h_i(x)=h_i(-x)$. }

Suppose $f \in A$ has more than two degrees of freedom. Then there are continuous functions $h_1, h_2, h_3$ (supported in $[-L,L]$) and $\delta > 0$ such that for all $\varepsilon_1, \varepsilon_2, \varepsilon_3 \in [-\delta,\delta]$ the function $f+\varepsilon_1 h_1 + \varepsilon_2 h_2 + \varepsilon_3 h_3$ is log-concave (note that these function are not necessarily contained in $A$). The space of solutions $\varepsilon_1, \varepsilon_2, \varepsilon_3$ to the system of equations
\begin{align*}
	\varepsilon_1 \int h_1 + \varepsilon_2 \int h_2 + \varepsilon_3 \int h_3 & = 0 \\
	\varepsilon_1 \int |x|^ph_1 + \varepsilon_2 \int |x|^ph_2 + \varepsilon_3 \int |x|^ph_3 & = 0
\end{align*}
is of dimension at least $1$. Therefore this space intersected with the cube $[-\delta,\delta]^3$ contains a symmetric interval and, in particular, two antipodal points $(\eta_1, \eta_2, \eta_3)$ and $-(\eta_1, \eta_2, \eta_3)$. Take $f_+= f+ \eta_1 h_1 + \eta_2 h_2 + \eta_3 h_3$ and $f_-= {\red f}- \eta_1 h_1 - \eta_2 h_2 - \eta_3 h_3$, which are both in $A$. Then, $f=\frac12(\tilde f_+ + \tilde f_-)$ and therefore $f$ is not an extremal point.

\bigskip
\noindent\underline{Step IV.} Densities with at most $2$ degrees of freedom are \emph{simple}.

We want to determine all nonincreasing log-concave functions $f$ on $[0,\infty)$ with degree of freedom at most $2$. Suppose $x_1<x_2< \ldots< x_n$ are points of differentiability of the potential $V=-\log {\red f}$, such that $0<V'(x_1)<V'(x_2) < \ldots < V'(x_n)$. Define 
\[
V_i(x) = \left\{ 
\begin{array}{ll}
	V(x), & x < x_i \\
	V(x_i)+(x-x_i)V'(x_i), & x \geq x_i.
\end{array} \right. 
\]
We claim that $e^{-V}(1+ \delta_0 + \sum_{i=1}^n \delta_i V_i )$ is a log-concave non-increasing function for $|\delta_i| \leq \varepsilon$, with $\varepsilon$ sufficiently small. To prove log-concavity we observe that on each interval the function is of the form $e^{-V(x)}(1+\tau_1+ \tau_2 x + \tau_3 V(x))$. On the interval $[0,x_1]$ it is of the form $e^{-V}(1+ {\red \tau_1 + \tau_2 V})$. Log-concavity follows from Lemma 1 in \cite{FG06}. We also have to ensure that the density is nonincreasing. On $[0,x_1]$ it follows from the fact that 
\[
	V'-(\log(1+\tau V))'= V'\cdot\left( 1 - \frac{\tau }{1+\tau V} \right) \geq 0
\]
for small $\tau$. On the other intervals we have similar expressions
\[
	V'(x) - \frac{\tau_2 + \tau_3 V'(x)}{1+\tau_1+ \tau_2 x + \tau_3 V(x)} > 0,
\]
which follows from the fact that $V'(x) > \alpha$ for some $\alpha > 0$.

From this it follows that if there are points $x_1<x_2< \ldots< x_n$, such that $0<V'(x_1)<V'(x_2) < \ldots < V'(x_n)$, then $e^{-V}$ has degree of freedom $n+1$. It follows that the only function with degree of freedom at most $2$ is of the form
\[
V(x) = \left\{ 
\begin{array}{ll}
	\beta, & x < a \\
	\beta+\gamma(x-a), & x \in [a,a+b].
\end{array} \right. 
\]

\subsection*{A two-point inequality}

It remains to show that for every density $f$ of the form
\[
f(x) = c\1_{[0,a]}(|x|) + ce^{-\gamma(|x|-a)}\1_{[a,a+b]}(|x|),
\]
where $c$ is a positive normalising constant and $a, b$ and $\gamma$ are nonnegative, we have
\[
\Ent{f} - \log \sigma_p(f) \geq \log\left[2(p+1)^{1/p}\right]
\]
with equality if and only if $f$ is uniform. If either $b$ or $\gamma$ are zero, then $f$ is a uniform density and we directly check that there is equality. Therefore let us from now on assume that both $b$ and $\gamma$ are positive and we shall prove the strict inequality.
Since the left-hand side does not change when $f$ is replaced by $x \mapsto \lambda f(\lambda x)$ for any positive $\lambda$, we shall assume that $\gamma = 1$. Then the condition $\int f = 1$ is equivalent to $2c(a + 1 - e^{-b})=1$. We have
\begin{align*}
\Ent{f} &= -2ac\log c - 2\int_0^b ce^{-x}\log(ce^{-x}) \dd x \\
&= -2c(a+1-e^{-b})\log c + 2c(1-(1+b)e^{-b}) \\
&= -\log c + \frac{1-(1+b)e^{-b}}{a+1-e^{-b}}.
\end{align*}
Moreover, 
\[
\sigma_p^p(f) = 2c\left(\frac{a^{p+1}}{p+1} + \int_0^b(x+a)^pe^{-x}\dd x\right).
\]
Putting these together yields
\begin{align*}
\Ent{f} - \log \sigma_p(f) - \log\left[2(p+1)^{1/p}\right] &= \frac{1-(1+b)e^{-b}}{a+1-e^{-b}} {\red +} \frac{p+1}{p}\log(a+1-e^{-b}) \\
&\quad- \frac{1}{p}\log\left[a^{p+1} + {\red (p+1)}\int_0^b(x+a)^pe^{-x}\dd x\right].
\end{align*}
Therefore, the proof of Theorem \ref{thm:ent-pmom} is complete once we show the following two-point inequality.

%

\begin{lemma}\label{lm:2point}
For nonnegative $s$, positive $t$ and $p \in (0,2]$ we have
\begin{align*}
\log\Big[s^{p+1} &+ (p+1)\int_0^t (s+x)^pe^{-x}\dd x\Big] < (p+1)\log[s+1-e^{-t}] + p\frac{1-(1+t)e^{-t}}{s+1-e^{-t}}.
\end{align*}
\end{lemma}
\begin{proof}
Integrating by parts, we can rewrite the left hand side as $\log[\int_0^t(s+x)^{p+1}\dd \mu(x)]$ for a Borel measure $\mu$ on $[0,t]$ (which is absolutely continuous on $(0,t)$ with density $e^{-x}$ and has the atom $\mu(\{t\}) = e^{-t}$). With $s$ and $t$ fixed, this is a strictly convex function of $p$ (by H\"older's inequality). The right hand side is linear as a function of $p$. Therefore, it suffices to check the inequality for $p=0$ and $p=2$. For $p=0$ the inequality becomes equality. For $p=2$, after computing the integral and exponentiating both sides, the inequality becomes
\[
s^3+3(1-e^{-t})s^2+6(1-(1+t)e^{-t})s + 3e^{-t}(2e^t-t^2-2t-2) < a^3e^{2\frac{b}{a}},
\]
where we put $a = s+1-e^{-t}$ and $b = 1-(1+t)e^{-t}$, which are positive. We lower-bound the right hand side using the estimate $e^{x} > 1+x+\frac12x^2+\frac16x^3$, $x > 0$, by
\begin{align*}
a^3\left(1 + 2\frac{b}{a} + 2\frac{b^2}{a^2} + \frac{4}{3}\frac{b^3}{a^3}\right) &= a^3 + 2a^2b + 2ab^2 + \frac{4}{3}b^3.
\end{align*}
Therefore it suffices to show that
\[
s^3+3(1-e^{-t})s^2+6(1-(1+t)e^{-t})s + 3e^{-t}(2e^t-t^2-2t-2) \leq a^3 + 2ba^2 + 2b^2a + \frac{4}{3}b^3.
\]
After moving everything on one side, plugging in $a$, $b$, expanding and simplifying, it becomes
\begin{align*}
2e^{-t}u(t)\cdot s^2 + e^{-2t}v(t)\cdot s + \frac{1}{3}e^{-3t}w(t) \geq 0,
\end{align*}
where
\begin{align*}
u(t) &= e^t - 1 -t,\\
v(t) &= 3e^{2t}-2te^t-12e^t+2t^2+8t+9,\\
w(t) &= e^{3t} + 3e^{2t}(3t^2-4t-13) + 3e^t(6t^2+20t+19) -4t^3-18t^2-30t-19.
\end{align*}
It suffices to prove that these functions are nonnegative for $t\geq 0$. This is clear for $u$. For $v$, we check that $v(0) = v'(0) = v''(0) = 0$ and
\[
\frac12 e^{-t}v'''(t) = 12e^t-t-9 \geq 12(t+1) - t - 9 = 11t + 3 \geq 3.
\]
For $w$, we check that $w'(0) = w''(0) = w'''(0) = 0$, $w^{(4)}(0) = 18$ and 
\begin{align*}
\frac13 e^{-t}w^{(5)}(t) &= 81e^{2t} + 32e^t(3t^2+11t-8) + 80t+6t^2+239 \\
&\geq 81e^{2t} - 8\cdot 32e^{t} + 239 = 81\left(e^t-\frac{128}{81}\right)^2 + \frac{2975}{81} \geq \frac{2975}{81}.
\end{align*}
It follows that $v(t)$ and $w(t)$ are nonnegative for $t \geq 0$.
\end{proof}

\begin{remark}\label{rem:neccond}
If we put $s = 0$ and $t \to \infty$ in the inequality from Lemma \ref{lm:2point}, we get $\log\Gamma(p+2) \leq p$ (in particular $p < 2.615$). We suspect that this necessary condition is also sufficient for the inequality to hold for all positive $s$ and $t$.
\end{remark}

%

\section{R\'enyi entropy minimizers}
\label{sec:renyi}

\subsection{A R\'enyi extension of Theorem~\ref{thm:ent-pmom}}
\label{ss:renyi-ext}

For $q\in(0,1)\cup(1,\infty)$, the {\it R\'enyi entropy of order $q$} of a probability density $f$ on $\R$ is defined as:
$$
h_q(f)=\frac{1}{1-q}\log\bigg(\int_{\R} f^q(x)\dd x\bigg).
$$
For $q=0, 1, \infty$, the entropies $h_q(f)$ are defined in a limiting sense. Thus
$$
h_1(f)=h(f)=-\int_{\R} f(x)\log f(x) \dd x 
$$
is the Shannon differential entropy; the R\'enyi entropy of order 0 is
$$
h_0(f)=\log|\text{supp}(f)| ,
$$
where $\text{supp}(f)$ is the support of $f$, defined as the closure of the set $\{x: f(x)>0\}$
and $|A|$ represents the Lebesgue measure of the subset $A$ of $\R$;
and the R\'enyi entropy of order $\infty$ is
$$
h_{\infty}(f)=-\log\|f\|_{\infty},
$$
where $\|f\|_{\infty}$ is the essential supremum of $f$ with respect to Lebesgue measure on $\R$.
It is an easy consequence of H\"older's inequality that the R\'enyi entropies of a fixed density $f$
are monotonically decreasing in the order: $h_q(f) \geq h_r(f)$ if $0\leq q\leq r\leq \infty$.
Moreover, if the density $f$ is log-concave, $\|f\|_{\infty}$ is just the maximum value of $f$ by continuity properties of convex functions;
also, the R\'enyi entropies of $f$ of all orders are necessarily finite, and can be bounded in terms of each other \cite{FMW16, FLM20}.

We have the following extension of Theorem~\ref{thm:ent-pmom} to R\'enyi entropies of orders between 0 and 1.

\begin{theorem}\label{thm:renyi-pmom}
Let $X$ be a symmetric log-concave random variable and $p \in (0,2]$. Then, for any $q\in [0,1]$,
\[
h_q(X) \geq \log \sigma_p(X) + \log\left[2(p+1)^{1/p}\right],
\]
with equality if and only if $X$ is uniformly distributed on a symmetric interval. Moreover, by taking the limit as $p\downarrow 0$,
\[
h_q(X) \geq \E (\log |X|) + \log (2e).
\]
\end{theorem}

\begin{proof}
The strict inequality holds for non-uniform measures by monotonicity of R\'enyi entropies in the order, and 
it is easily checked that equality holds for the uniform.
\end{proof}

Thus, in Theorem~\ref{thm:ent-pmom}, one can replace Shannon entropy by R\'enyi entropy of any order 
$q$ in $[0,1]$ and the same statement holds true. In fact, one can also use Theorem~\ref{thm:renyi-pmom} 
to get bounds on R\'enyi entropies of order greater than 1. In order to do this, we use the sharp R\'enyi 
entropy comparison result implicit in \cite{FMW16} and explicitly discussed in \cite{MW20}
(see also \cite[Corollary 7.1]{FLM20}), which states that if $f$ is a log-concave density in $\mathbb{R}^n$, then for $p\geq q>0$,
$$
h_q(f)-h_p(f)\leq n\frac{\log q}{q-1}-n\frac{\log p}{p-1},
$$
with equality achieved for the product density whose one-dimensional marginals are the symmetrized exponential distribution.
Consequently we may write, for $q>1, p\in (0,2]$, and in our setting of a random variable $X$ with a symmetric, log-concave distribution on $\R$, 
$$
h_q(X) 
\geq \log \sigma_p(X)+ \log\left[2(p+1)^{1/p}\right] -\frac{\log q}{q-1}.
$$
While this does provide a bound on arbitrary R\'enyi entropies in terms of moments (which is new to the best of our knowledge), we
emphasize that it is not sharp when $q>1$.

It is instructive to compare Theorem~\ref{thm:renyi-pmom} with results of Lutwak, Yang and Zhang \cite{LYZ05:1} on 
maximizing R\'enyi entropies subject to moment constraints
(the $p=2$ case was independently discovered by \cite{CHV03} and the $q=1$ case is classical, see, e.g., \cite{CT06:book}).
They showed that if $p$ and $\E |X|^p$ are fixed  positive numbers, 
and if 
\begin{equation}\label{eq:lyz-cond}
q>\frac{1}{1+p} \quad \text{(or equivalently } p> \frac{1}{q}-1 ), 
\end{equation}
then $h_q(X)$ is maximized by a scaling of a ``generalized standard Gaussian density'' of the form
$$
g_{p,q}(x)= 
\begin{cases}
A_{p,q}^{-1} \bigg(1+\frac{\beta}{p}|x|^p \bigg)^{-\frac{1}{1-q}} , &\text{ if } q< 1\\
A_{p,1}^{-1} \exp\{-\frac{|x|^p}{p}\}, &\text{ if } q= 1 .\\ 
\end{cases}
$$
Here, 
$$
\beta=
\frac{q}{1-q}-\frac{1}{p} 
$$
is well defined when $q<1$ (and always negative because of the assumed relationship \eqref{eq:lyz-cond}), 
and $A_{p,q}$ is a normalizing constant given by
\begin{equation}\label{eq:A}
A_{p,q}=
\frac{A_{p,1}}{\beta^{1/p}}\cdot \frac{\Gamma(\frac{1}{1-q}-\frac{1}{p})} {\Gamma(\frac{1}{1-q})}
\end{equation}
when $q<1$, 
and $A_{p,1}= 2p^{1/p} \Gamma(1+\frac{1}{p})$,
with $\Gamma(x):=\int_0^\infty t^{x-1}e^{-t} dt$ 
as usual denoting the Gamma 
function.

Define the {\it R\'enyi entropy power of order $q$} of $X$ by
$$
N_q(X)= \frac{1}{A_{2,q}^{2}e}  e^{2h_q(X)} .
$$
This normalization has not been used in the literature before, but we use it since it simplifies our expressions while being consistent with the usual
entropy power in the sense that $N_1(X)=N(X)$.
For a random variable $Z_{p,q}$ drawn from the density $g_{p,q}$, it turns out (see, e.g., \cite{LLYZ13} for a sketch of the computation) that $\sigma_p(Z_{p,q})=\E|Z_{p,q}|^p=1$ and
the maximum entropy power for random variables with $p$-th moment equal to 1 is given by
$$
N_q(Z_{p,q})= 
\begin{cases}
\frac{1}{e} \big(\frac{A_{p,q}}{A_{2,q}}  \big(1+\frac{\beta}{p}\big)^{\frac{1}{1-q}} \big)^{2}, &\text{ if } q< 1\\
\frac{1}{e} \big(\frac{A_{p,1}}{A_{2,1}} e^\frac{1}{p} \big)^{2}, &\text{ if } q= 1 .\\ 
\end{cases}
$$
Thus one has the following upper bound  for the R\'enyi entropy power of a random variable $X$
when $p\in (0,2]$ and $q\in (\frac{1}{1+p},1]$:
\[
N_q(X)\leq N_q(Z_{p,q}) \sigma_p(X)^2 .
\]
Note that the maximizers of R\'enyi entropy
(which are scaled versions of $Z_{p,q}$) are not always log-concave; for example, when $q=1$, it is easy to see from the formula above
that they are log-concave precisely when $p\geq 1$.

This may be compared to Theorem~\ref{thm:renyi-pmom}, which may be written in the form
$$
N_q(X)\geq \frac{4}{A_{2,q}^{2} e} [(p+1) ]^{\frac{2}{p}} \sigma_p(X)^2 ,
$$
when $p\in (0,2], q\in (0,1]$, and $X$ is symmetric and log-concave. In particular, for $p=2$,
we obtain that for any $q\in (1/3,1)$, we have the following sandwich bound when $X$ is symmetric and log-concave:
\begin{equation}\label{eq:Nq-sw}
\frac{12}{A_{2,q}^{2} e}\leq \frac{N_q(X)}{\Var(X)}\leq  \frac{1}{e} \bigg(1+\frac{\beta}{2}\bigg)^{\frac{2}{1-q}} .
\end{equation}

\subsection{Implication for relative $q$-entropy}

As we did in Section~\ref{sec:slicing} for the case of $q=1$, it is possible to express Theorem~\ref{thm:renyi-pmom}
as a bound on a kind of distance between a symmetric, log-concave distribution and the generalized Gaussian with the same
$p$-th moment. In order to do this, we need to define the notion of relative $q$-entropy, whose properties were first systematically studied
by Ashok Kumar and Sundaresan \cite{AS15:1}. The {\it relative $q$-entropy} between densities $f$ and $u$ is defined as
$$
I_{\red q}(f\|u)=\frac{q}{1-q}\log \int \frac{f}{\|f\|_q} \bigg(\frac{u}{\|u\|_q}\bigg)^{q-1} ,
$$
when $q\in (0,1)\cup (1,\infty)$; as pointed out in  \cite{AS15:1}, the relative $q$-entropy is genuinely a notion of distance
between densities rather than between probability measures since it may depend on the reference measure being used.
There is a way to write the  relative $q$-entropy in terms of more familiar notions of distance. 
Define the {\it R\'enyi divergence of order $\alpha$} between densities $f$ and $g$
by
$$
D_\alpha(f\|u)=\frac{1}{\alpha-1}\log \int f^\alpha u^{1-\alpha} 
$$
for $\alpha\in (0,1)\cup(1,\infty)$; by taking limits, it is clear that $D_1(f\|g)$ should be defined as the usual relative entropy $D(f\|g)$.
Also define the {\it $\alpha$-escort density} of a density $f$ by
$$
f_\alpha(x)=\frac{f^\alpha(x)}{\int f^\alpha} .
$$
Then $I_q(f\|u)=D_{1/q}(f_q\|u_q)$ (see \cite[Lemma 2]{AS15:1}), which also makes clear that $I_1(f\|u)=D_1(f\|u)=D(f\|u)$.

The following proposition is a particular example of general facts about relative $q$-entropy projections onto linear families
of probability measures that were proved in  \cite{AS15:1}.

\begin{proposition}\label{prop:as}\cite[Corollary 13]{AS15:1}
Suppose $q\in (0,1]$, and let $\mathcal{P}$ be the family of probability measures such that
the mean of the function $T:\R\rightarrow\R$ under them is fixed at a particular value $t$. Let 
the random variable $X$ have a distribution from $\mathcal{P}$, and let $Z$ be a random variable that maximizes 
the R\'enyi entropy of order $q$ over $\mathcal{P}$. Then
$$
I_q(X\|Z)=h_q(Z)-h_q(X).
$$
\end{proposition}

There continues to be a relation between the two sides of the identity when $q>1$ but the equality is replaced by an inequality
in this case \cite{AS15:1}; we do not, however, use that observation in this note since we only consider $q\leq 1$.

Clearly, combining Proposition~\ref{prop:as} with Theorem~\ref{thm:renyi-pmom} allows us to write the latter
as a bound on the relative $q$-entropy from a generalized Gaussian density.

\begin{corollary}\label{cor:Iq}
Let  $X$ be a random variable with a symmetric, log-concave distribution. Then, for $p\in (0,2]$ and $q\in (\frac{1}{1+p},1)$,
and $Z$ being the multiple of $Z_{p,q}$ that has the same $p$-th moment as $X$, we have
$$
I_q(X\|Z)\leq \log \bigg[A_{p,q}\frac{(1+\frac{\beta}{p})^{\frac{1}{1-q}}}{2 (p+1)^\frac{1}{p}} \bigg] ,
$$
with equality if and only if $X$ is uniformly distributed on a symmetric interval.
\end{corollary}

\subsection{Reverse entropy power inequalities}

The entropy power inequality {\red (see \eqref{eq:EPI})} has spawned a large literature, both in mathematics
due to its fundamental connections to geometric functional inequalities, and in engineering
due to its many applications in quantifying the fundamental limits of various communication systems. 
Some recent refinements of the entropy power inequality may be found, e.g., in \cite{MB06:isit, MB07, MG19, MNT20}.

One may formally strengthen it by using the
invariance of entropy under affine transformations of determinant $\pm 1$, 
i.e., $ \mathcal{N} (AX) =  \mathcal{N} (X)$ whenever $|{\rm det}(A)|=1$. Specifically,
\begin{equation}\label{epi-aff}
\inf_{A_1, A_2}  \mathcal{N} (A_1 X+A_2 Y) \geq  \mathcal{N}(X) +  \mathcal{N}(Y),
\end{equation}
where the matrices $A_1$ and $A_2$ range over $SL(n,\R)$, i.e., over 
entropy-preserving linear transformations. 
It was shown by \cite{BM11:cras} that {\red inequality \eqref{epi-aff}}
can be reversed with a constant independent of dimension 
if we restrict to log-concave distributions. More precisely,
there exists a universal constant $C$ 
such that if $X$ and $Y$ are independent random vectors 
in $\R^n$ with log-concave densities, 
\[\label{eq:repi}
\inf_{A_1, A_2}  \mathcal{N} (A_1 X+A_2 Y)\, \leq\, C\, \big[ \mathcal{N} (X) +  \mathcal{N} (Y)\big],
\]
where $A_1$ and $A_2$ range over $SL(n,\R)$.
This reverse entropy power inequality is analogous
to Milman's \cite{Mil86} reverse Brunn-Minkowski inequality, 
which is a celebrated result in convex geometry.
Thus the reverse entropy power inequality of \cite{BM11:cras} (and  its extension to larger classes 
of ``$s$-concave measures'' in \cite{BM12:jfa}) can be seen
as an extension of the analogies between geometry and information theory
(discussed, for example, in \cite{CC84, DCT91, Gar02, MMX17:0, FMMZ18}).  

The universal constant in the reverse entropy power inequality of \cite{BM11:cras} is not explicit. 
However, explicit constants are known when further assumptions of symmetry are made. 
For example, Cover and Zhang \cite{CZ94} (cf., \cite{MK18}) showed that 
if $X$ and $Y$ are (possibly dependent) random vectors in
$\R^n$, with the same
log-concave marginal density, then $h(X+Y)\leq h(2X)$.
In particular, for i.i.d. random vectors $X, X'$ with a log-concave distribution, the 
reverse entropy power inequality holds with 
both linear transformations being the identity, and with a universal constant of 2:
$\mathcal{N} (X+X')\leq  \mathcal{N} (2X)= 4 \mathcal{N} (X)=2[ \mathcal{N} (X)+ \mathcal{N} (X')]$.
%
%
%
%

There has been much recent interest in developing lower bounds for the R\'enyi entropies of
convolutions, which may be thought of as ``R\'enyi entropy power inequalities''. While the growing
literature on the subject is surveyed in \cite{MMX17:0}, the only orders for which sharp inequalities
are known are $q=0$ (which corresponds to the Brunn-Minkowski inequality),
$q=1$ (which corresponds to the original Shannon-Stam entropy power inequality),
and $q=\infty$ (which corresponds to
generalizations of Rogozin's inequality for convolution that were only developed recently \cite{MMX17:1}).

While suboptimal forms of R\'enyi entropy power inequalities that hold for general densities are known for $q\in (1,\infty)$ (see, e.g., \cite{Li18:1}), the only
known inequalities for $q\in (0,1)$ were recently obtained in \cite{MM19, LMM20} under the assumption
that the densities being convolved are log-concave (or more generally, $s$-concave). 

Our results imply a reverse R\'enyi entropy power inequality for orders $q\in (\frac{1}{3},1]$.

\begin{corollary}\label{cor:repi}
Let  $X, Y$ be uncorrelated random variables with symmetric, log-concave distributions. Then
$$
\mathcal{N}(X+Y)\leq \frac{\pi e}{6} [\mathcal{N}(X)+\mathcal{N}(Y)] .
$$
Furthermore, if $q\in (\frac{1}{3},1)$, then
$$
\mathcal{N}_q(X+Y)\leq  \frac{A_{2,q}^{2}}{12} \bigg(1+\frac{\beta}{2}\bigg)^{\frac{2}{1-q}} [\mathcal{N}_q(X)+\mathcal{N}_q(Y)] ,
$$
where {\red the constant $A_{2,q}$ is defined in \eqref{eq:A}}.
\end{corollary}

\begin{proof}
We now observe that Theorem~\ref{thm:ent-pmom} 
easily gives us an explicit reverse entropy power inequality for one-dimensional 
symmetric log-concave random variables. Indeed, 
$$
\mathcal{N}(X+Y)\leq \Var(X+Y)
= \Var(X)+\Var(Y)
\leq  \frac{\pi e}{6} [\mathcal{N}(X)+\mathcal{N}(Y)] ,
$$
as long as $X$ and $Y$ are uncorrelated.

Using the inequality \eqref{eq:Nq-sw}, we write
\begin{equation*}\begin{split}
\mathcal{N}_q(X+Y)&\leq \frac{1}{e} \bigg(1+\frac{\beta}{2}\bigg)^{\frac{2}{1-q}} \Var(X+Y)\\
&= \frac{1}{e} \bigg(1+\frac{\beta}{2}\bigg)^{\frac{2}{1-q}} \big[\Var(X)+\Var(Y)\big]\\
&\leq \frac{1}{e} \bigg(1+\frac{\beta}{2}\bigg)^{\frac{2}{1-q}}  \frac{A_{2,q}^{2} e}{12} [\mathcal{N}_q(X)+\mathcal{N}_q(Y)] ,
\end{split}\end{equation*}
as long as $X$ and $Y$ are uncorrelated.
\end{proof}

Under the additional assumption of central symmetry,  the first inequality of Corollary~\ref{cor:repi}
 improves a result of \cite{MK18b}, who showed that $\mathcal{N}(X+Y)\leq \frac{\pi e}{2} [\mathcal{N}(X)+\mathcal{N}(Y)]$
for uncorrelated, log-concave random variables $X, Y$.
Other reverse entropy power inequalities for centrally symmetric, log-concave random vectors, 
motivated by analogies to Busemann's theorem in convex geometry,
are discussed in \cite{BNT16}.


\section{An aside on the capacities of additive noise channels}
\label{sec:cvx-cap}

We now make a general observation (we will comment in Section~\ref{sec:disc} on its connection to 
the rest of this paper).
Consider the mutual information $I(X;Y)$ as a function of an input distribution (the distribution of $X$,
which we assume to have a density $p(x)$ with respect
to some reference measure on the input space) and a channel or Markov kernel $K(x,dy)$ 
that represents the behavior of $Y$ conditioned on $X$ (which we assume to have a density
with respect to a reference measure on the output space, so that we may write it as $W(x,y) dy$).
It is well known that $\tilde{I}(p,W)=I(X;Y)$ is concave in $p$ for fixed $W$,
and convex in $W$ for fixed $p$. 

%
%

Suppose the input and output spaces are the same set $G$, and $G$ has a group structure induced
by a binary operation $+$. We will also require the group to have a locally compact, Polish topology on it
compatible with the group operation (i.e., $x+y$ is a continuous function of $(x,y)$),
so that the group has a translation-invariant Haar measure and
there are no measure-theoretic complications with conditional densities.
Choose the channel $W$ given by $W_u(x,y)=u(y-x)$ with $u$ being a probability density function
(all densities are taken with respect to the Haar measure of $G$).
Then the joint distribution of $(X,Y)$ has
a density given by $p(x) u(y-x)$, and $Y$ represents the output of an additive noise channel with
input $X$ and noise $N\sim u$. If we restrict to the world of additive noise channels and
write the mutual information between input and output as $\bar{I}(p,u)=\tilde{I}(p, W_u)$,
then the convexity of $\tilde{I}(p,W)$ in $W$ translates to the convexity of $\bar{I}(p,u)$
in $u$. By Shannon's channel coding theorem, the capacity of the additive noise channel 
with noise density $u$ is given by
$$
C(u)=\sup_{p \in \mathcal{P}}  \bar{I}(p,u) ,
$$
where $\mathcal{P}$ is the class of permissible input distributions.
Since $C$ is a supremum of convex functions (of $u$), and a supremum of a class of convex functions
is always convex (as easily verified by looking, for example, at sublevel sets), we deduce the following
basic fact. (Although very simple, to our surprise, this fact does not seem to have been observed 
before as far as we could tell.)

\begin{theorem}\label{thm:cap-cvx}
The functional  $u\mapsto C(u)$, assigning to the noise density of an additive noise channel 
the capacity under any fixed set of input constraints, 
is convex.
\end{theorem}

This is a very general statement-- we emphasize that the input/output space is an arbitrary locally compact Polish group,
and the class $\mathcal{P}$ of permissible input distributions is also arbitrary.

Specializing to the case where the input and output spaces are just $\mathbb{R}^d$ with
usual addition, we note that this immediately gives us a bound on the capacity of channels where the additive noise
can be represented as a mixture. 

\begin{corollary}\label{cor:gaus-mix}
Suppose the noise $N$ has a density $u$ on $\mathbb{R}^d$ that is a scale mixture of Gaussians, i.e., 
$$
u(x)= \sum_{i=1}^M \alpha_i g_{\Sigma_i} (x) ,
$$
where we use $g_\Sigma$ to denote a centered Gaussian density with covariance matrix $\Sigma$.
Then the capacity $C_P(N)$ of the additive noise channel with
noise $N$ and input power constraint $P$ satisfies
\begin{equation}\label{eq:gaus-mix-1st-bd}
C_P(N)\leq  \sum_{i=1}^M \alpha_i C_P(Z_i) ,
\end{equation}
where $Z_i\sim g_{\Sigma_i}$.
\end{corollary}
 
 Recall that the capacity $C(Z_\Sigma)$ of an additive noise channel with Gaussian noise $Z$
 of covariance $\Sigma$ has an explicit formula given by
\begin{equation}\label{eq:gaus-cap}
 C_P(Z_\Sigma)= \max_{\text{tr}(\tilde{\Sigma})\leq P } \frac{1}{2} \log \bigg[\frac{\det(\tilde{\Sigma}+\Lambda)}{\det(\Lambda)}\bigg],
\end{equation}
where $\Lambda$ is the diagonal matrix of eigenvalues that appears in the decomposition 
$\Sigma=U\Lambda U^T$ with $U$ being an 
orthogonal matrix. Moreover, the maximum in the expression for the capacity is achieved 
when $\tilde{\Sigma}$ is diagonal with eigenvalues determined by spectral water-filling.

Note that the expression that is being maximized in \eqref{eq:gaus-cap}
may be rewritten as 
$$
\frac{1}{2} \log \det(I+\tilde{\Sigma} \Lambda^{-1}) ,
$$
by using the fact that the determinant of a product of matrices is the product of
the determinants. It is well known that for any fixed positive-semidefinite $\tilde{\Sigma}$,
this is a convex function of $\Lambda$ (see, e.g., \cite{Kim19} and references therein).
Therefore, if we restrict to diagonal matrices $\Sigma$, then
$C_P(Z_\Sigma)$ is the pointwise maximum of a collection of convex functions,
and hence itself convex in $\Sigma$.

Let us compare the bound from Corollary~\ref{cor:gaus-mix} with the one that 
follows from Proposition~\ref{prop:ihara} in the special case where all the 
covariance matrices $\Sigma_i$ are diagonal. The latter is given by
\begin{equation}\label{eq:gaus-mix-2nd-bd}
C_P(N)\leq  C_P(Z_{\Sigma}) + D(N\|Z_{\Sigma}) ,
\end{equation}
where $\Sigma= \sum_{i=1}^M \alpha_i \Sigma_i$ is the covariance matrix of $N$.
The convexity of $C_P(Z_\Sigma)$ in $\Sigma$ observed in the preceding paragraph
implies that the first term of \eqref{eq:gaus-mix-2nd-bd} is better than 
the bound obtained in \eqref{eq:gaus-mix-1st-bd}; on the other hand,
the second term is \eqref{eq:gaus-mix-2nd-bd} is making it worse.
Thus, in general, the capacity bounds obtained from Corollary~\ref{cor:gaus-mix} 
and Proposition~\ref{prop:ihara} seem to not be comparable\footnote{Note that if we 
further estimate $D(N\|Z)$ using the convexity of relative entropy, we 
can get an upper bound with an explicit expression in \eqref{eq:gaus-mix-2nd-bd}
since the relative entropy between two Gaussian distributions has a closed form.
}.

It would be interesting to explore if there are also inequalities that go in the other direction
to the convexity of capacity. Specifically, in the case of entropy, it is well known that although
entropy is concave, one also has $h(\sum_{i=1}^k p_i f_i) \leq \sum_{i=1}^k p_i h(f_i) + H(p)$,
where $p=(p_1, \ldots, p_k)$ is a discrete probability distribution and $H$ is the discrete Shannon 
entropy. Indeed various refinements of this elementary inequality are known (see, e.g., \cite{MTBMS20},
where mixtures of log-concave distributions are a particular focus).
One wonders if it is possible to obtain, in a similar spirit, constraints on how convex (as a function of the noise distribution)
the channel capacity is.



\section{Discussion}
\label{sec:disc}

A significant consequence (Corollary~\ref{cor:cap}) of our main theorem in this paper is that
any additive noise channel with symmetric, log-concave noise
has a capacity that is at most 0.254 bits per channel use
greater than the capacity of an additive Gaussian noise  
channel with the same noise power. This begs the following question:

\vspace{.2cm}
\noindent{\bf Question 1}: Among all symmetric, log-concave noise distributions with fixed variance, 
which are the distributions that maximize the capacity of an additive noise channel?
identify the maximizers of $C_P(N)$ as one varies $N$ over all noise distributions with fixed variance.
In short, which is the ``best''  symmetric, log-concave noise?
\vspace{.2cm}

We now make various remarks, several of which shed light on why Question 1 is a 
non-trivial question:

\begin{enumerate}

\item Even for very simple and special noise distributions such as the uniform distribution, the capacity of the corresponding additive noise channel
remains unknown! As well known and explained for example in \cite{RM14}, a closed-form formula exists when $A/\Delta$ is an integer if we consider a system where
input amplitude (rather than power) is constrained by $A$ and the noise is uniform on $[-\Delta,\Delta]$, but this says nothing about the case
when one constrains variance rather than peak amplitude. 

\item There is a stream of research in information theory studying capacity-achieving distributions and especially conditions under which these are discrete (see, e.g., \cite{FA18:1}),
and also some studying dependence on system parameters \cite{EPK17:isit}. However, none of these appear to help with the problem at hand.
Indeed, they do not say much even for specific channels such as the uniform noise channel that are clearly candidates for the extremizer, 
since they typically make global assumptions such as positivity of the noise density on the entire real line (this is necessary for much of the 
complex analytic machinery that is typically used in the subject because analytic extensions to a strip are obviously impossible for something like the uniform density).

\item  By Theorem~\ref{thm:cap-cvx}, the capacity of an additive noise channel is a convex function
of the noise distribution. Thus, modulo analytical details, the problem of maximizing the capacity over a class of permissible
distributions (in our case, symmetric, log-concave noise distributions with fixed variance)
is related to the determination of the extreme points of the convex hull of this class.

\item Using the localization technique as outlined in Section~\ref{sec:proof-ent-vs-var}, it should be possible to show that
the best log-concave noise under  a power budget is of the form $e^{-V}$, where $V$ is two-piece affine. 
Unfortunately, given the difficulty of computing the
capacity even for uniform noise, computing capacities corresponding to this family of noises is out of reach,
and thus  the observation does not seem to be of much help with Question 1.

\item We have considered three optimization problems over the class of
symmetric, log-concave noise distributions with fixed variance in this paper:
minimizing entropy (or maximizing relative entropy), maximizing capacity,
and maximizing the isotropic constant. 
Whereas the first two optimization problems both involve maximizing a convex
function, we note that the third does not. Indeed, 
consider the family of densities on $\R$ given by
$f_t= t \frac{e^{-|x|}}{\red 2} + (1-t)\frac{1_{[-a,a]}(x)}{2a}$
for $t\in (0,1)$, i.e., mixtures of  a symmetric exponential distribution
and the uniform distribution on $[-a,a]$. 
For a symmetric density $f$ on $\R$,
the isotropic constant is just 
$f(0)\sqrt{\Var(X)}$ (with $X\sim f$),
which gives
$$
\ell(t):= L_{f_t} = \sqrt{2t + \frac{a^2}{3} (1-t)} \bigg[ \frac{t}{2} + \frac{(1-t)}{2a} \bigg] .
$$
For example, when $a = 1.4$, the second derivative $\ell''(t)$ changes sign,
and so the isotropic constant is not a convex functional.

\item It is interesting to note that, already in dimension 1, the extremizers for the isotropic constant formulation of the slicing problem
are different from those for the relative entropy formulation of it. This discrepancy is explained by the form of Theorem~\ref{thm:D-iso}:
if $L_f$ had been maximized at the uniform distribution, then clearly $D(f)$ would have been as well, or if $D(f)$ had been maximized at
the symmetrized exponential, then $L(f)$ would have been as well. But both of these are one-directional statements and both premises
are false.

\item The upper bound in Proposition~\ref{prop:ihara} is sharp if and only if $Y$ is Gaussian, in which case $D(Y)=0$ and clearly not maximized. 
So our identification of the extremizer for the relative entropy cannot possibly identify the extremizer for Question 1.

\end{enumerate}

There are other interesting questions related to those we have explored that also remain open. 
First, we expect that if the symmetry assumption in our Theorem~\ref{thm:ent-pmom}
is replaced by the assumption that $X$ has zero mean, the extremal distribution should be the exponential distribution on a half-line
(at least for $p=2$). This is strongly suggested by analytical calculations and computational experiments done by Liyao Wang
and the first-named author, but we have been unable to prove it.

Second, we mention the question of maximizing $h(X-X')-h(X)$, where $X'$ is an independent and identically distributed
copy of $X$, over all log-concave distributions on $\R^n$. Note that the objective function here is invariant with respect to
scaling of $X$, so no power constraint is needed. This question is the entropy
analogue of the Rogers-Shepherd inequality in convex geometry, and the best known bound, proved in \cite{BM13:goetze}, 
is $h(X-X')-h(X)\leq n$. However, a sharp bound has not been proved even in dimension 1. It is conjectured in \cite{MK18}
that the extremal distribution for this problem is also the exponential distribution.

Third, it would be interesting to find the extremal distributions that minimize and maximize
$$
\frac{h(X-X')-h(X)}{h(X+X')-h(X)}
$$
over all log-concave distributions on $\R^n$. This objective function (which, once again, is invariant to scaling of $X$)
measures the asymmetry of the distribution of $X$-- note that it must equal 1 if the density of $X$ is
symmetric. The answer is unknown even in dimension 1. It was proved in \cite{MK10:isit, KM14} that the ratio above must lie
in the interval $[\frac{1}{2},2]$, for arbitrary (not necessarily log-concave) distributions on $\R^n$. 
(In fact, this result turns out to hold in remarkable generality-- on any locally compact abelian group with entropies defined with
respect to the Haar measure-- as shown in \cite{MK18}.) One expects even tighter bounds for log-concave distributions.

Finally, there are many open questions about capacity-achieving distributions. 
For example, if the noise $Z$ is symmetric and log-concave, must the capacity-achieving distribution
(i.e., the maximizer of $h(X+Z)$ under a variance constraint on $X$) be log-concave? 
If the answer is yes, then combining with results on the concentration of information \cite{BM11:aop, FMW16, FLM20}, one 
should be able to obtain finite-blocklength bounds for channels with symmetric log-concave noise.
While we are unable to answer the question about capacity-achieving distributions, we note that
Corollary~\ref{cor:cap-loss} does tell us that using a symmetric log-concave input distribution can only cause a very limited loss;
if the answer to the question at the beginning of this paragraph is affirmative, then Corollary~\ref{cor:cap-loss} would be replaced
by the identity $C_P^{LC}(N)=C_P(N)$.


\end{document}